\documentclass[11pt, a4paper]{article}

\usepackage[T1]{fontenc}
\usepackage[utf8]{inputenc}
\usepackage{lmodern}         

\usepackage{etoolbox}

\usepackage[english]{babel}

\usepackage{geometry}
\usepackage{mathtools}       
\usepackage{amsfonts}
\usepackage{esint}           
\usepackage{amssymb}        
\usepackage{isomath}         

\newcommand\numberthis{\addtocounter{equation}{1}\tag{\theequation}}

\usepackage{amsthm}
\theoremstyle{plain}
\newtheorem{theorem}{Theorem}[section]
\newtheorem{lemma}[theorem]{Lemma}

\newtheorem{proposition}[theorem]{Proposition}

\theoremstyle{definition}

\theoremstyle{remark}
\newtheorem{remark}{Remark}

\RequirePackage{natbib}
\bibpunct{(}{)}{;}{a}{,}{,}

\usepackage{natbib}
\bibpunct[, ]{(}{)}{;}{a}{}{,}

\usepackage{url}
\usepackage[hidelinks]{hyperref}

\usepackage{xcolor}

\usepackage{graphicx}
\usepackage[tableposition=top]{caption}
\usepackage{subcaption}  
\usepackage{longtable}
\usepackage{tabu}
\usepackage{booktabs}

\newlength{\origtabcolsep}
\usepackage{relsize}

\usepackage{tabto}

\usepackage[title, page]{appendix}

\let\originalleft\left
\let\originalright\right
\renewcommand{\left}{\mathopen{}\mathclose\bgroup\originalleft}
\renewcommand{\right}{\aftergroup\egroup\originalright}

\DeclareMathOperator{\E}{E}
\DeclareMathOperator{\Var}{Var}
\DeclareMathOperator{\Cov}{Cov}
\newcommand*{\p}{p}
\newcommand*{\diff}{\mathop{}\kern -2pt\mathrm{d}}

\usepackage{bbm}

\newcommand*{\Model}{\ensuremath{M}}
\newcommand*{\elpd}[2]{{\mathrm{elpd}\bigr(#1 \mid #2\bigl)}}
\newcommand*{\elpdHat}[2]{{\widehat{\mathrm{elpd}}_\mathrm{\scriptscriptstyle LOO}\bigr(#1 \mid #2\bigl)}}

\newcommand*{\elpdHati}[3]{{\widehat{\mathrm{elpd}}_{\mathrm{\scriptscriptstyle LOO},\, #3}\bigr(#1 \mid #2\bigl)}}
\newcommand*{\seHatName}[3]{{\widehat{\mathrm{SE}}_{#3}\bigr(#1 \mid #2\bigl)}}
\newcommand*{\ShortElpdHat}{{L}} 

\makeatletter
\def\@maketitle{%
  \begin{center}%
  \let \footnote \thanks
     {\large {\bf \@title\vspace{\baselineskip}} \par}%
     {\normalsize
       \begin{tabular}[t]{c}%
         \@author
       \end{tabular}\par}%
     {\small {\vspace{.1in} \@date}}%
  \end{center}%
}
\makeatother

\title{Unbiased estimator for the variance of the leave-one-out cross-validation estimator for a Bayesian normal model with fixed variance}

\author{%
Tuomas Sivula\thanks{Department of Computer Science, Aalto University, Finland.}
\and
Måns Magnusson\thanks{Department of Statistics, Uppsala University, Sweden. Most of the work was done while at Aalto University.}
\and
Aki Vehtari\footnotemark[1]
}

\date{15th February 2022}

\makeatletter
\hypersetup{
  pdftitle={Unbiased estimator for the variance of a LOO-CV estimator for a Bayesian normal model with fixed variance},
  pdfauthor={%
 Tuomas Sivula,
 Måns Magnusson,
 Aki Vehtari}
}
\makeatother

\providecommand{\keywords}[1]
{
  \small	
  \textbf{\textit{Keywords:}} #1
}

\begin{document}
\maketitle

\begin{abstract}
When evaluating and comparing models using leave-one-out cross-validation (LOO-CV), the uncertainty of the estimate is typically assessed using the variance of the sampling distribution.
Considering the uncertainty is important, as the variability of the estimate can be high in some cases.
An important result by \citet{bengio_Grandvalet_2004} states that no general unbiased variance estimator can be constructed, that would apply for any utility or loss measure and any model.
We show that it is possible to construct an unbiased estimator considering a specific predictive performance measure and model.
We demonstrate an unbiased sampling distribution variance estimator for the Bayesian normal model with fixed model variance using the expected log pointwise predictive density (elpd) utility score. This example demonstrates that it is possible to obtain improved, problem-specific, unbiased estimators for assessing the uncertainty in LOO-CV estimation.
%
\end{abstract}


\section{Introduction}
\label{sec_intro}

Leave-one-out cross-validation (LOO-CV) is a popular method for estimating the predictive performance of Bayesian models based on new, unseen, data with respect to some utility or loss function.
As discussed by \citet[][Section 5.2.1]{arlot_celisse_2010_cv_survey}, the variability of the LOO-CV estimator can be high in some cases, and it is important to consider the uncertainty of the estimate.
In order to assess this uncertainty, one would typically estimate the standard error~\citep[see for example][]{Vehtari+Lampinen:2002,Vehtari+Ojanen:2012,Vehtari+Gelman+Gabry:2017_practical,Yao2018stacking,vehtari2019loopkg}.
An important result by \citet{bengio_Grandvalet_2004} states that no unbiased estimator exists for the variance of the sampling distribution.
While unbiasedness as such is not necessary for a useful estimator, experimental results show that in some typical cases, the current way of estimating the uncertainty can lead to underestimating the variance up to a factor of two or more~\citep{bengio_Grandvalet_2004,varoquaux2017166,varoquaux201868}. Furthermore, the bias in this estimator is theoretically unbounded~\citep{bengio_Grandvalet_2004}.
An alternative approach for assessing the uncertainty of a LOO-CV estimator is to use Bayesian bootstrap~\citep{Rubin:1981a,Vehtari+Lampinen:2002} but this method shares the same issues as estimating the sampling distribution variance. As demonstrated by \citet{sivula_2020_loo_uncertainty} in the context of model comparison, problematic situations include comparing models with similar predictions, model misspecification, and small data sets.
By improving the estimated uncertainty of the LOO-CV estimate, it could be possible to make a more robust assessment of the predictive performances of models and their differences.

While no unbiased estimator for the variance of a LOO-CV sampling distribution can be constructed in general, we show that it is possible to construct such an estimator by considering the specific problem setting at hand.
Previously, the variance estimation has been analysed in a model- and measure agnostic way by considering the problem given the obtained pointwise LOO-CV estimates.
However, given the data and the model, one could apply the model structure to the LOO-CV estimator with the selected measure to directly find the variance.
We show, as an example, that it is possible to find an unbiased variance estimator for the LOO-CV expected log pointwise predictive density (elpd) under a simple Bayesian normal model with fixed data variance. The results indicate the possibility of deriving other problem-specific estimators that could have a negligible bias or otherwise reduced error compared to the naive approach.

In the following, we first introduce the problem setting in Section~\ref{sec_problem_setting}.
In Section~\ref{sec_method}, as a demonstrational example, we present an improved estimator for the variance of the LOO-CV sampling distribution in the setting of a Bayesian normal model with fixed data variance.
Furthermore, in Section~\ref{sec_experiment}, we apply the derived improved estimator in a couple of simulated settings to assure it is unbiased and to illustrate the bias in the naive estimator.
Finally, Section~\ref{sec_conclusions} concludes the work by highlighting the findings and discussing their implications and possibilities for future research.

\subsection{Problem setting}\label{sec_problem_setting}

Consider data $y=(y_1, y_2, \dots, y_n)$ and let $\p_\text{true}(y)$ be the distribution representing the data generating mechanism.
For evaluating the predictive performance of a model $\Model$ conditional on an observed data set $y^\text{obs}$, we apply the \emph{expected log pointwise predictive density} (elpd) utility score~\citep{Vehtari+Ojanen:2012,Vehtari+Gelman+Gabry:2017_practical}:
\begin{align}
\label{eq_elpd}
    \elpd{\Model}{y^\text{obs}} &= \sum_{i=1}^n \int \p_\text{true}(y_i) \log \p_\Model\big(y_i \mid y^\text{obs}\big) \diff y_i
\,,
\end{align}
where $\log \p_\Model(y_i \mid y^\text{obs})$ is the posterior predictive log density for the model $\Model$ fitted for the data set $y^\text{obs}$.
The LOO-CV estimate for $\elpd{\Model}{y^\text{obs}}$ is
\begin{align*}
\elpdHat{\Model}{y^\text{obs}} &= \sum_{i=1}^n \elpdHati{\Model}{y^\text{obs}}{i}
    \numberthis\label{eq_elpd_loo}
\,,
\end{align*}
where
\begin{align*}
\elpdHati{\Model}{y^\text{obs}}{i} &= \log \p_\Model\big(y^\text{obs}_i \big| y^\text{obs}_{-i}\big) \\
    &= \log \int \p_\Model\big(y^\text{obs}_i \big| \theta \big)\p_\Model \big(\theta \big| y^\text{obs}_{-i}\big) \diff \theta
    \numberthis
\end{align*}
is the leave-one-out predictive log density for the $i$th observation $y^\text{obs}_i$ using model $\Model$, given all the other observations denoted with $y^\text{obs}_{-i}$.

For estimating the uncertainty about the estimand in a LOO-CV estimate, the commonly used naive approach is to estimate the variance of the estimator $\elpdHat{\Model}{y}$ by~\citep{vehtari2019loopkg}
\begin{align*}
\seHatName{\Model}{y^\text{obs}}{\text{naive}}^2 = \frac{n}{n-1}\sum_{i=1}^n \Bigg(\elpdHati{\Model}{y^\text{obs}}{i} - \frac{1}{n}\sum_{j=1}^n \elpdHati{\Model}{y^\text{obs}}{j} \Bigg)^2
\,, \numberthis\label{eq_naive_estimator}
\end{align*}
which is based on the incorrect assumption that the terms $\elpdHati{\Model}{y}{i}$ are independent \cite[see][for a discussion on the different fold covariance structures]{bengio_Grandvalet_2004}.
Assuming the observations $y_i$ are i.i.d., the bias of this estimator is $-n^2 \gamma$, where $\gamma = \Cov\Bigl(\elpdHati{\Model}{y}{i}, \allowbreak \elpdHati{\Model}{y}{j} \Bigr)$ for any $i \neq j$~\citep{sivula_2020_loo_uncertainty}.

In this work, we seek an unbiased estimator for the variance $\operatorname{Var}\left( \elpdHat{\Model}{y} \right)$ in the context of one specific model $\Model$: the Bayesian normal model with fixed variance. We approach the problem by utilising the model's known predictive density function $p_\Model\big(\tilde{y} \big| y\big)$. Instead of the pointwise LOO-CV terms $\elpdHati{\Model}{y}{i}$, we consider the variance directly as a function of the data $y$ and derive the variance analytically. Based on this, we combine various moment estimators of $y$ to construct an unbiased estimator for the target variance. We do not need to utilise the inappropriate assumption of the independence of the pointwise LOO-CV terms $\elpdHati{\Model}{y}{i}$. Instead, as discussed in Section~\ref{sec_method}, only reasonable assumptions for $y$ suffices for obtaining the unbiased estimator.

\section{Unbiased variance estimator for a normal model}\label{sec_method}

In this section, we show that it is possible to construct an unbiased estimator for the variance of the sampling distribution of the LOO-CV estimator in a specific case.
We apply the LOO-CV with elpd utility score to estimate the predictive performance of a Bayesian normal model with fixed data variance.

Considering the true data generating mechanism, we assume that the first four moments exist for $\p_\text{true}(y)$ and that
the observations $y_i$ are independent.
This can be summarized as:
\begin{align*}
    \E[y_i] &= \mu \,,
    \\
    \Var(y_i) &= \sigma^2 \,,
    \\
    \E[(y_i - \E[y_i])^r] &= \mu_r\,, \quad r=3,4, &&\text{($r$th central moment)},
    \\
    \E[f(y_i)g(y_j)] &= \E[f(y_i)]\E[g(y_j)] &&\text{(independence)}
    \label{eq_assumptions}\numberthis
\end{align*}
for $i,j=1,2,\dots,n$, $i \neq j$, and for all functions $f,g : \mathbb{R} \rightarrow \mathbb{R}$ for which the expectations $\E[f(y_i)]$ and $\E[g(y_j)]$ exists.
In addition, we assume $n \geq 4$.
We do not assume that the observations are identically distributed.

Considering the applied model, likelihood is
\begin{align}
    y_i \mid \theta, \Model \sim \operatorname{N}\left(\theta, \sigma_\mathrm{m}^2 \right)
    \,,\label{eq_model}
\end{align}
where $\Model$ indicates the applied model, $\theta$ is the sole estimate model parameter, and $\sigma_\mathrm{m}^2$ is a fixed data variance parameter.
The prior distribution for $\theta$ is
\begin{align}
    \theta \sim \operatorname{N}\left(0, \sigma_0^2 \right)
    \,, \label{eq_prior}
\end{align}
where $\sigma_0^2$ is a fixed prior variance parameter.
The fixed data variance parameter $\sigma_\mathrm{m}^2$ reflects how the model considers the magnitude of the variability of the data. A fixed $\sigma_\mathrm{m}^2$ is mainly chosen to simplify derivations. The fixed prior variance parameter $\sigma_0^2$ reflects the prior belief of the magnitude of the variability of the unknown mean parameter. The fixed model parameters $\sigma_\mathrm{m}^2$ and $\sigma_0^2$ can be chosen freely.

The model and the true data generating mechanism have different assumptions.
The notation $\Model$ in the conditional arguments in Equation~\eqref{eq_model} emphasises that the relation reflects the applied model, not the true data generating mechanism.
Considering the analysis of the LOO-CV estimation, the model may be misspecified so that it represents the true data generating mechanism poorly.
In particular, in Equation~\eqref{eq_model}, the evaluated model assumes the observations are identically distributed, while such an assumption is not made about the true data generating mechanism $\p_\text{true}(y)$.

In the following, we derive the variance of the sampling distribution $\elpdHat{\Model}{y}$ in Lemma~\ref{lemma_var_of_elpd}
and show that it is possible to estimate the required terms in Lemma~\ref{lemma_estims_for_some_moment_products}.
Finally, in Proposition~\ref{proposition_result}, we show that it is possible to construct an unbiased estimator for the variance of the sampling distribution $\elpdHat{\Model}{y}$ for the normal model defined in equations~\eqref{eq_model} and~\eqref{eq_prior}.

\begin{lemma}
\label{lemma_var_of_elpd}
Let the data generating mechanism for $y=[y_1,y_2,\dots,y_n]$ be such that Equation~\eqref{eq_assumptions} holds and let the model $\Model$ be as defined in equations~\eqref{eq_model} and~\eqref{eq_prior}. We then have
\begin{align*}
    \Var\left(\elpdHat{\Model}{y}\right)
        &= 4 n (a + b + c)^2 \mu^2 \sigma^2 \\
&\quad
    +\left(
        - n a^2
        + \frac{2n}{n-1} b^2
        + \frac{n (2 n - 3) (n - 3)}{(n - 1)^3} c^2
        - \frac{2n}{n-1} a c
        + \frac{4n(n-2)}{(n-1)^2} b c
    \right) \sigma^4 \\
&\quad
    +\frac{4 n (a + b + c) (a (n - 1) + c)}{n - 1} \mu \mu_3 
\\&\quad
    +\left(
        n a^2
        + \frac{n}{(n-1)^2} c^2
        + \frac{2n}{n-1} a c
    \right) \mu_4
    \,,\numberthis
    \label{eq_true_var_in_lemma}
\end{align*}
where
\begin{align}
    a &= - \frac{1}{2} \frac{ \sigma_\mathrm{m}^2 + (n-1) \sigma_0^2}{\sigma_\mathrm{m}^2(\sigma_\mathrm{m}^2 + n \sigma_0^2)} \,, \\
    b &= \frac{(n-1) \sigma_0^2}{\sigma_\mathrm{m}^2 (\sigma_\mathrm{m}^2 + n \sigma_0^2)} \,, \\
    c &= - \frac{1}{2} \frac{(n-1)^2 \sigma_0^4}{\sigma_\mathrm{m}^2(\sigma_\mathrm{m}^2+(n-1)\sigma_0^2)(\sigma_\mathrm{m}^2+n \sigma_0^2)} \,,
\end{align}
\end{lemma}
\begin{proof}
\renewcommand{\qedsymbol}{} 
See Appendix.
\end{proof}

\begin{lemma}
\label{lemma_estims_for_some_moment_products}
Let the data generating mechanism for $y=[y_1,y_2,\dots,y_n]$ be such that Equation~\eqref{eq_assumptions} holds.
Let
\begin{equation}
\widehat{\alpha}_k = \frac{1}{n}\sum_{i=1}^n y_i^k
\end{equation}
be the $k$th sample raw moment of the data and
\begin{align*}
    \widehat{\mu^4} = \binom{n}{4}^{-1} \sum_{i_1 \neq i_2 \neq i_3  \neq i_4} y_{i_1}y_{i_2}y_{i_3}y_{i_4}
   \,, \numberthis
\end{align*}
where the summation is over all possible combinations of $i_1, i_2, i_3, i_4 \in \{1,2,\dots,n\}$ without repetition,
be an unbiased estimator for the fourth power of the mean.
Now
\begingroup
\allowdisplaybreaks
\begin{align*}
    \widehat{\mu^2 \sigma^2} &= \frac{- n^3 \widehat{\alpha}_1^4 + 2 n^3 \widehat{\alpha}_2 \widehat{\alpha}_1^2 - 4 (n - 1) n \widehat{\alpha}_3 \widehat{\alpha}_1 - (2 n^2 - 3n) \widehat{\alpha}_2^2 + 2 (2 n - 3) \widehat{\alpha}_4}{2 (n - 3) (n - 2) (n - 1)} - \frac{1}{2} \widehat{\mu^4}
    \,. \numberthis
    \label{eq_estim_ms}
\\
    \widehat{\sigma^4} &= \frac{
        n^3 \widehat{\alpha}_1^4 - 2 n^3 \widehat{\alpha}_2 \widehat{\alpha}_1^2 + (n^3 - 3 n^2 + 3n) \widehat{\alpha}_2^2 + 4 n (n - 1) \widehat{\alpha}_3 \widehat{\alpha}_1 + n (1 - n) \widehat{\alpha}_4
    }{
        (n - 3) (n - 2) (n - 1)
    }
    \,, \numberthis \label{eq_estim_s4}
\\
    \widehat{\mu \mu_3} &= \frac{
        - 2 (n^2 + n -3) \widehat{\alpha}_4
        - 6 n^3 \widehat{\alpha}_1^2 \widehat{\alpha}_2
        + n (6 n -9) \widehat{\alpha}_2^2
        + 3 n^3 \widehat{\alpha}_1^4
        + 2 n^2 (n + 1) \widehat{\alpha}_1 \widehat{\alpha}_3
    }{
        2 (n-3) (n-2) (n-1)
    }
    + \frac{1}{2} \widehat{\mu^4}
    \,, \numberthis
    \label{eq_estim_mm}
\intertext{and}
    \widehat{\mu_4} &= \frac{
        -3n^4 \widehat{\alpha}_1^4
        +6n^4 \widehat{\alpha}_1^2\widehat{\alpha}_2
        +(9-6n)n^2 \widehat{\alpha}_2^2
        +(-12+8n-4n^2)n^2 \widehat{\alpha}_1\widehat{\alpha}_3
        +(3n-2n^2+n^3)n \widehat{\alpha}_4
    }{
        (n-3)(n-2)(n-1)n
    }
     \numberthis \label{eq_estim_m4}
\end{align*}
\endgroup
are unbiased estimators for the parameters $\mu^2 \sigma^2$, $\sigma^4$, $\mu \mu_3$, and $\mu_4$ respectively.
\end{lemma}
\begin{proof}
\renewcommand{\qedsymbol}{} 
See Appendix.
\end{proof}

\begin{proposition}
\label{proposition_result}
Let the data generating mechanism for $y=[y_1,y_2,\dots,y_n]$ be such that Equation~\eqref{eq_assumptions} holds and let the model $\Model$ be as defined in equations~\eqref{eq_model} and~\eqref{eq_prior}. Then there exist an unbiased estimator for $\Var\left(\elpdHat{\Model}{y}\right)$.
\end{proposition}
\begin{proof}
The required variance $\Var\left(\elpdHat{\Model}{y}\right)$, derived in Equation~\eqref{eq_true_var_in_lemma} in Lemma~\ref{lemma_var_of_elpd}, is a linear combination of the terms $\mu^2 \sigma^2$, $\sigma^4$, $\mu \mu_3$, and $\mu_4$, for which the multipliers depend on the known sample size $n$ and fixed parameters $\sigma_\mathrm{m}^2$ and $\sigma_0^2$.
Unbiased estimators for each of these terms are presented in Lemma~\ref{lemma_estims_for_some_moment_products}.
Thus, it is possible to construct an unbiased estimator for the required variance $\Var\left(\elpdHat{\Model}{y}\right)$ by substituting the terms in Equation~\eqref{eq_true_var_in_lemma} with the respective unbiased estimators in Lemma \ref{lemma_estims_for_some_moment_products}.
\end{proof}

\begin{remark}
The numerical stability of the estimator is  dependent on the moment estimators presented in Lemma~\ref{lemma_estims_for_some_moment_products}. In some settings, the estimator could be modified to consider the data skewness and/or excess kurtosis to be zero in order to reduce the variability caused by estimating them.
\end{remark}

\section{Simulated experiment}
\label{sec_experiment}

Here we compare the expectations of the naive LOO-CV variance estimator presented in Equation~\eqref{eq_naive_estimator} and the unbiased estimator discussed in Proposition~\ref{proposition_result} in a simulated experiment.
The motivation of the experiment is two-fold. First, to verify that the presented improved estimator is unbiased, and second, to highlight the bias of the naive estimator.
By choosing different parameters, three different  problem setting cases are studied:
\begin{enumerate}
\item well-matching data,
\item under-dispersed data, and
\item under-dispersed, skewed, and heavy-tailed data.
\end{enumerate}

We utilise Monte Carlo (MC) sampling for analysing the expectations of the variance estimators.
We simulate 20\,000 independent data sets of size $n=16$ under a known data generating mechanism and compute both the naive and the improved LOO-CV variance estimates for each of them.
The model is defined in equations~\eqref{eq_model} and~\eqref{eq_prior}.
The fixed model parameters are set to $\sigma^2_\mathrm{m}=1.2^2$, $\sigma^2_0=2^2$.
Based on the obtained variance estimates, we apply the Bayesian bootstrap (BB, \citealp{Rubin:1981a}) to infer the MC uncertainty of the expectations of the estimators.
We use bootstrap sample size of 4\,000 and Dirichlet distribution parameter $\alpha=1$.
Compared to using the MC standard error, we use the BB approach to capture any skewness in the MC uncertainty.
The expectations are analysed under three different data generating mechanisms:
\begin{enumerate}
\item $y_i \sim \mathrm{N}(0, 1.2)$. The model matches the data well. The naive estimator underestimates the variance.
\item $y_i \sim \mathrm{N}(2, 0.1)$. The data is under-dispersed. The naive estimator overestimates the variance.
\item $y_i \sim \operatorname{skew-N}(\mathrm{location}=-2,\, \mathrm{scale}=0.16,\, \mathrm{shape}=10)$. The data is under-dispersed, skewed, and heavy-tailed. The naive estimator underestimates the variance.
\end{enumerate}

\begin{figure}[!htb]
\centerline{\includegraphics[width=0.9\textwidth]{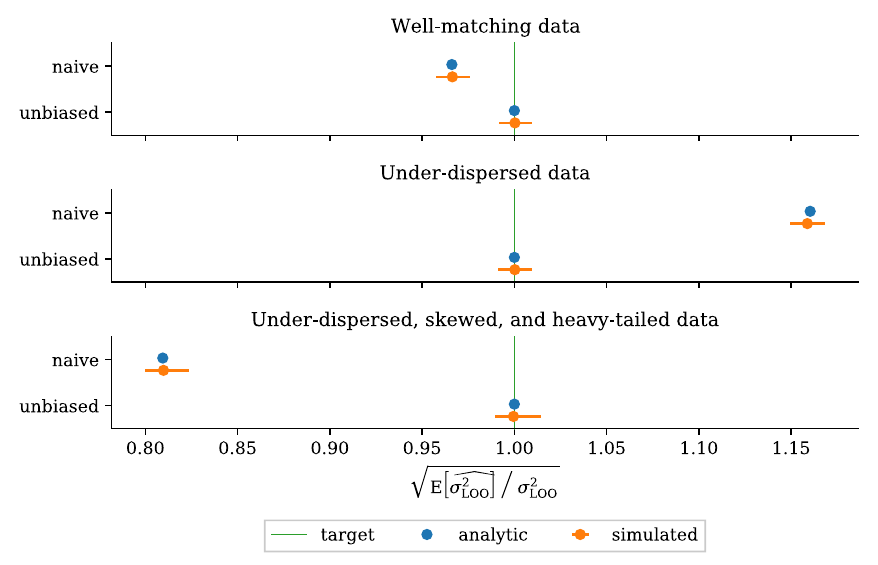}}
\caption{%
The expectation of the naive and unbiased LOO-CV variance estimators $\widehat{\sigma^2_\mathrm{\scriptscriptstyle LOO}}$ estimated using Bayesian bootstrap (BB, \citealp{Rubin:1981a}) in a simulated experiment under three different data generating mechanisms: well-matching, under-dispersed, and under-dispersed skewed heavy-tailed data respectively.
The x-axis is transformed to
the square root of the ratio to the LOO-CV estimator's true variance $\sigma^2_\mathrm{\scriptscriptstyle LOO}$.
The analytic expectations (blue) match the simulated results (yellow) in all cases.
The BB uncertainty is illustrated using a dot and a line corresponding to the mean and 95 \% credible interval, respectively.
The naive estimator underestimates or overestimates the variance while the improved estimator discussed in Proposition~\ref{proposition_result} is unbiased.
}\label{fig_app_experiment_results}
\end{figure}

The results of the experiment are illustrated in Figure~\ref{fig_app_experiment_results}.
The analytic target variance and the expectation of the unbiased estimator are obtained using the equations presented in Lemma~\ref{lemma_var_of_elpd}.
The analytic expectation of the naive estimator is calculated by applying equations~\eqref{eq_app_var_looi} and \eqref{eq_app_cov_looi_looj} from the appendix, in the following known expectation in a more general setting~\citep{sivula_2020_loo_uncertainty}:
\begin{align}
    \E\left[\seHatName{\Model}{y}{\text{naive}}^2\right] = n \Var\Bigl(\elpdHati{\Model}{y}{i} \Bigr) - n \Cov\Bigl(\elpdHati{\Model}{y}{i}, \elpdHati{\Model}{y}{j} \Bigr)
    \,.
\end{align}
The simulated results match with the analytic ones in all the experiment settings.
As expected, the estimated expectation of the unbiased estimator is around the target variance.
Depending on the situation, the naive estimator underestimates or overestimates the variance.
The source code of the experiment is available online~\citep{sivula_2020_loocv_uncertainty_github_repo}.

\section{Discussion}\label{sec_conclusions}

The current common way of estimating the uncertainty in the LOO-CV model assessment and comparison utilizes a naive biased estimator for the variance of the sampling distribution.
The naive approach may result in a significantly underestimated variability~\citep{bengio_Grandvalet_2004,varoquaux2017166,varoquaux201868} and bad calibration of the uncertainty~\citep{sivula_2020_loo_uncertainty}.
An important result by \citet{bengio_Grandvalet_2004} states that no unbiased variance estimator can be constructed in general, that would apply for any utility or loss measure and any model.
We show that it is possible to construct an unbiased estimator considering a specific predictive performance measure and model.

While the unbiasedness itself is not necessary for a feasible estimator of the variance of the sampling distribution in this context, the presented result serves as an example of the existence of such estimators and as an example of a possibility to improve over the naive approach.
We expect that this approach of finding problem-specific estimators would extend to other, more complex models and benefit the data analysis field with a more accurate assessment of the uncertainty in the widely used LOO-CV model comparison.
Although deriving the variance of the sampling distribution could be unfeasible in some problem settings, it could be possible to apply some approximative method to obtain an estimator, not necessarily unbiased one, that would result in a better calibration of the uncertainty than with the naive variance estimator.
Further research is needed to study the possibility of extending the problem-specific approach to more complex settings.


\section*{Acknowledgements}
We acknowledge the computational resources provided by the Aalto Science-IT project. This work was supported by the Academy of Finland grants (298742 and 313122) and Academy of Finland Flagship programme: Finnish Center for Artificial Intelligence FCAI.





\bibliography{references}



\section{Appendices}
\appendix

\section{Proofs for the lemmas}

In this appendix, we present proofs for lemmas~\ref{lemma_var_of_elpd} and~\ref{lemma_estims_for_some_moment_products}
in sections~\ref{sec_app_var_elpd} and~\ref{sec_app_samp_estims_for_terms} respectively.

\subsection{Proof of Lemma~\ref{lemma_var_of_elpd}}\label{sec_app_var_elpd}
In this section, we give proof for Lemma~\ref{lemma_var_of_elpd} by deriving the LOO-CV variance estimator as a function of the data given the Normal model.
Let us first restate the lemma.
\begin{quote}
Let the data generating mechanism for $y=[y_1,y_2,\dots,y_n]$ be such that Equation~\eqref{eq_assumptions} holds and let the model $\Model$ be as defined in equations~\eqref{eq_model} and~\eqref{eq_prior}. We have
\begin{align*}
    &\Var\left(\elpdHat{\Model}{y}\right)
        \\
&\qquad= 4 n (a + b + c)^2 \mu^2 \sigma^2 \\
&\qquad\qquad+
    \left(
        - n a^2
        + \frac{2n}{n-1} b^2
        + \frac{n (2 n - 3) (n - 3)}{(n - 1)^3} c^2
        - \frac{2n}{n-1} a c
        + \frac{4n(n-2)}{(n-1)^2} b c
    \right) \sigma^4 \\
&\qquad\qquad+
    \frac{4 n (a + b + c) (a (n - 1) + c)}{n - 1} \mu \mu_3 \\
&\qquad\qquad+
    \left(
        n a^2
        + \frac{n}{(n-1)^2} c^2
        + \frac{2n}{n-1} a c
    \right) \mu_4
    \,,\numberthis
\end{align*}
where
\begin{align}
    a &= - \frac{1}{2} \frac{ \sigma_\mathrm{m}^2 + (n-1) \sigma_0^2}{\sigma_\mathrm{m}^2(\sigma_\mathrm{m}^2 + n \sigma_0^2)} \,, \\
    b &= \frac{(n-1) \sigma_0^2}{\sigma_\mathrm{m}^2 (\sigma_\mathrm{m}^2 + n \sigma_0^2)} \,, \\
    c &= - \frac{1}{2} \frac{(n-1)^2 \sigma_0^4}{\sigma_\mathrm{m}^2(\sigma_\mathrm{m}^2+(n-1)\sigma_0^2)(\sigma_\mathrm{m}^2+n \sigma_0^2)} \,,
\end{align}
\end{quote}
\begin{proof}
Let $\sum_{k\neq S}$ denote $\sum_{k \in \{1,2,\dots,n\}\setminus S}$
and let $i,j,h,k \in \{1,2,\dots,n\}, i \neq j, i \neq h, i \neq k, j \neq h, j \neq k, h \neq k$, that is $i,j,h,k$ are all distinct.
We make the following assumptions:
\begin{align*}
    \E[y_i] &= \mu \,,
    \\
    \Var(y_i) &= \sigma^2 \,,
    \\
    \E[(y_i - \E[y_i])^r] &= \mu_r\,, \quad r=3,4, &&\text{($r$th central moment)},
    \\
    \E[f(y_i)g(y_j)] &= \E[f(y_i)]\E[g(y_j)] &&\text{(independence)}
    \numberthis
\end{align*}
for all functions $f,g : \mathbb{R} \rightarrow \mathbb{R}$ for which the expectations $\E[f(y_i)]$ and $\E[g(y_j)]$ exists.
In addition, we assume $n \geq 4$.
Let $\overline{y} = \frac{1}{n}\sum_{p=1}^n y_p$.
Given the assumptions, we have
\begin{align*}
\E[y_i^2] &= \E[y_i]^2 + \Var(y_i) = \mu^2 + \sigma^2 \,,\\
\E[y_i^3] &= \mu_3 + 3 \sigma^2 \mu + \mu^3 \,,\\
\E[y_i^4] &= \mu_4 + 4 \mu_3 \mu + 6 \sigma^2 \mu^2 + \mu^4 \,,\\
\E[y_i^2]^2 &= \left( \mu^2 + \sigma^2 \right)^2 = \mu^4 + 2 \mu^2 \sigma^2 + \sigma^4
    \numberthis\,.
\end{align*}
Let
\begingroup
\allowdisplaybreaks
\begin{gather}
A = \sum_{p \neq \{i,j\}} y_p \,, \qquad
B = \sum_{p \neq \{i,j\}} y_p^2 \,, \qquad
C = \sum_{p \neq \{i,j\}} \sum_{q \neq \{i,j,p\}} y_p y_q \,.
\end{gather}
With these, we have the following expectations:
\begin{align*}
\E[A] &= \E\left[ \sum_{p \neq \{i,j\}} y_p \right] = (n-2)\E[y_i]
    \,,\numberthis\\
\E[B] &= \E\left[ \sum_{p \neq \{i,j\}} y_p^2 \right] = (n-2)\E[y_i^2]
    \,,\numberthis\\
\E[C] &= \E\left[ \sum_{p \neq \{i,j\}} \sum_{q \neq \{i,j,p\}} y_p y_q \right]
    = (n-2)(n-3)\E[y_i]^2
    \,,\numberthis\\
\E[A^2] &= \E\left[ \left( \sum_{p \neq \{i,j\}} y_p \right)^2 \right]
    = \sum_{p \neq \{i,j\}} \sum_{q \neq \{i,j\}} \E[y_p y_q]
    = \sum_{p \neq \{i,j\}} \E[y_p^2] + \sum_{p \neq \{i,j\}} \sum_{q \neq \{i,j,p\}} \E[y_p y_q] \\
  &= (n-2)\E[y_i^2] + (n-2)(n-3)\E[y_i]^2
  \,,\numberthis\\
\E[B^2] &= \E\left[ \left( \sum_{p \neq \{i,j\}} y_p^2 \right)^2 \right]
    = \sum_{p \neq \{i,j\}} \sum_{q \neq \{i,j\}} \E[y_p^2 y_q^2]
    = \sum_{p \neq \{i,j\}} \E[y_p^4] + \sum_{p \neq \{i,j\}} \sum_{q \neq \{i,j,p\}} \E[y_p^2 y_q^2] \\
  &= (n-2)\E[y_i^4] + (n-2)(n-3)\E[y_i^2]^2
  \,,\numberthis\\
\E[C^2] &= \E\left[ \left( \sum_{p \neq \{i,j\}} \sum_{q \neq \{i,j,p\}} y_p y_q \right)^2 \right]
    = \E\left[ \sum_{p \neq \{i,j\}} \sum_{q \neq \{i,j,p\}} y_p y_q \sum_{r \neq \{i,j\}} \sum_{s \neq \{i,j,r\}} y_r y_s \right] \\
  &= \E\Bigg[ \sum_{p \neq \{i,j\}} \sum_{q \neq \{i,j,p\}} y_p y_q \Bigg(
        2 \bigg( y_p y_q + y_p \sum_{r \neq \{i,j,p,q\}} y_r + y_q \sum_{r \neq \{i,j,p,q\}} y_r \bigg) \\
        &\qquad\qquad + \sum_{r \neq \{i,j,p,q\}} \sum_{s \neq \{i,j,p,q,r\}} y_r y_s
  \Bigg) \Bigg] \\
  &= (n-2)(n-3)\bigg(2\Big(\E[y_i^2]^2 + 2 (n-4) \E[y_i^2]\E[y_i]^2 \Big) + (n-4)(n-5)\E[y_i]^4 \bigg)
  \,,\numberthis\\
\E[AB] &= \E\left[ \left( \sum_{p \neq \{i,j\}} y_p \right) \left( \sum_{p \neq \{i,j\}} y_p^2 \right) \right]
    = \sum_{p \neq \{i,j\}} \sum_{q \neq \{i,j\}} \E[y_p y_q^2] \\
    &= \sum_{p \neq \{i,j\}} \E[y_p^3] + \sum_{p \neq \{i,j\}} \sum_{q \neq \{i,j,p\}} \E[y_p y_q^2] \\
    &= (n-2)\E[y_i^3] + (n-2)(n-3)\E[y_i]\E[y_i^2]
    \,,\numberthis\\
\E[AC] &= \E\left[ \left( \sum_{p \neq \{i,j\}} y_p \right) \left( \sum_{p \neq \{i,j\}} \sum_{q \neq \{i,j,p\}} y_p y_q \right) \right]
    = \E\left[ \sum_{p \neq \{i,j\}} \sum_{q \neq \{i,j,p\}} y_p y_q \sum_{r \neq \{i,j\}} y_r \right] \\
  &= \E\left[ \sum_{p \neq \{i,j\}} \sum_{q \neq \{i,j,p\}} y_p y_q
    \left( y_p + y_q + \sum_{r \neq \{i,j,p,q\}} y_r \right) \right] \\
  &= (n-2)(n-3)\left( 2 \E[y_i^2]\E[y_i] + (n-4)\E[y_i]^3 \right)
  \,,\numberthis\\
\E[BC] &= \E\left[ \left( \sum_{p \neq \{i,j\}} y_p^2 \right) \left( \sum_{p \neq \{i,j\}} \sum_{q \neq \{i,j,p\}} y_p y_q \right) \right]
    = \E\left[ \sum_{p \neq \{i,j\}} \sum_{q \neq \{i,j,p\}} y_p y_q \sum_{r \neq \{i,j\}} y_r^2 \right] \\
  &= \E\left[ \sum_{p \neq \{i,j\}} \sum_{q \neq \{i,j,p\}} y_p y_q
    \left( y_p^2 + y_q^2 + \sum_{r \neq \{i,j,p,q\}} y_r^2 \right) \right] \\
  &= (n-2)(n-3)\left( 2 \E[y_i^3]\E[y_i] + (n-4)\E[y_i]^2\E[y_i^2] \right)
  \,.\numberthis
\end{align*}
\endgroup 
Now we can derive the following expectations, which are utilized later on:
\makeatletter
\setbool{@fleqn}{true}
\makeatother
\begingroup
\allowdisplaybreaks
\begin{align*}
\overline{y}_{-i}
&= \frac{1}{n-1} \sum_{p \neq \{i\}} y_p = \frac{1}{n-1} \left( y_j + \sum_{p \neq \{i,j\}} y_p \right)
=  \frac{1}{n-1} \left( y_j + A \right)
\numberthis \,,
\\
\E\left[\overline{y}_{-i}\right]
&= \frac{1}{n-1} (n-1) \E[y_i] = \mu
\numberthis \,,
\end{align*}

\begin{align*}
\overline{y}_{-i}\;\overline{y}_{-j}
&= \left( \frac{1}{n-1} \sum_{p \neq \{i\}} y_p \right)
    \left( \frac{1}{n-1} \sum_{p \neq \{j\}} y_p \right) \\
& = \frac{1}{(n-1)^2} \left(
        y_i y_j + y_i \sum_{p \neq \{i,j\}} y_p + y_j \sum_{p \neq \{i,j\}} y_p
        + \sum_{p \neq \{i,j\}} y_p^2
        + \sum_{p \neq \{i,j\}} \sum_{q \neq \{i,j,p\}} y_p y_q
    \right) \\
& = \frac{1}{(n-1)^2} \Big(y_i y_j + y_i A + y_j A + B + C \Big)
\numberthis
\\
\E\left[\overline{y}_{-i} \; \overline{y}_{-j}\right]
&= \frac{1}{(n-1)^2} \left(\E[y_i]^2 + 2 \E[y_i]\E[A] + \E[B] + \E[C] \right) \\
&= \frac{1}{(n-1)^2} \left(\E[y_i]^2 + 2 \E[y_i](n-2)\E[y_i] + (n-2)\E[y_i^2] + (n-2)(n-3)\E[y_i]^2 \right) \\
&= \frac{1}{(n-1)^2} \left( ((n-2)(n-3) + 2(n-2) + 1)\E[y_i]^2 + (n-2)\E[y_i^2] \right) \\
&= \mu^2 + \frac{n-2}{(n-1)^2} \sigma^2
\numberthis \,,
\end{align*}

\begin{align*}
\E\left[y_i y_j \overline{y}_{-i} \; \overline{y}_{-j}\right]
    &= \frac{1}{(n-1)^2} \E\left[y_i^2 y_j^2 + y_i^2 y_j A + y_j^2 y_i A + y_i y_j B + y_i y_j C \right] \\
    &= \frac{1}{(n-1)^2} \Big(
        \E[y_i^2]^2 + 2 \E[y_i^2] \E[y_i] \E[A] + \E[y_i]^2 \E[B] + \E[y_i]^2 \E[C]
    \Big) \\
    &= \frac{1}{(n-1)^2} \Big(
        \E[y_i^2]^2
        + 2 \E[y_i^2] \E[y_i] (n-2)\E[y_i]
        + \E[y_i]^2 (n-2)\E[y_i^2] \\
        &\qquad\qquad+ \E[y_i]^2 (n-2)(n-3)\E[y_i]^2
    \Big) \\
    &= \frac{1}{(n-1)^2} \Big(
        \E[y_i^2]^2
        + 3(n-2) \E[y_i^2] \E[y_i]^2
        + (n-2)(n-3)\E[y_i]^4
    \Big) \\
    &= \frac{1}{(n-1)^2} \Big(
        \mu^4 + 2 \mu^2 \sigma^2 + \sigma^4
        + 3(n-2) (\mu^2 + \sigma^2) \mu^2
        + (n-2)(n-3)\mu^4
    \Big) \\
    &=
        \mu^4
        + \frac{3n-4}{(n-1)^2} \mu^2 \sigma^2
        + \frac{1}{(n-1)^2} \sigma^4
\numberthis \,,
\end{align*}

\begin{align*}
\overline{y}_{-i}^2
&= \left( \frac{1}{n-1} \sum_{p \neq \{i\}} y_p \right)^2
    = \frac{1}{(n-1)^2} \sum_{p \neq \{i\}} \sum_{q \neq \{i\}}y_p y_q \\
&= \frac{1}{(n-1)^2} \left(
    y_j^2
    + 2 y_j \sum_{p \neq \{i,j\}} y_p
    + \sum_{p \neq \{i,j\}} y_p^2
    + \sum_{p \neq \{i,j\}} \sum_{q \neq \{i,j,p\}} y_p y_q
\right) \\
&= \frac{1}{(n-1)^2} \left(y_j^2 + 2 y_j A + B + C\right)
\numberthis \,, \\
\E\left[\overline{y}_{-i}^2\right]
&= \frac{1}{(n-1)^2} \left(\E[y_i^2] + 2 \E[y_i] \E[A] + \E[B] + \E[C]\right)
\\
&= \frac{1}{(n-1)^2} \left(
    \E[y_i^2] + 2 \E[y_i](n-2)\E[y_i] + (n-2)\E[y_i^2] + (n-2)(n-3)\E[y_i]^2 \right)
\\
&= \frac{1}{(n-1)^2} \left(
    (n-1) \E[y_i^2] + (n-1)(n-2) \E[y_i]^2 \right)
\\
&= \frac{1}{(n-1)^2} \left(
    (n-1) (\mu^2 + \sigma^2) + (n-1)(n-2) \mu^2 \right)
\\
&= \frac{1}{(n-1)^2} \left(
    (n-1)^2 \mu^2 + (n-1) \sigma^2 \right)
\\
&= \mu^2 + \frac{1}{n-1} \sigma^2
\numberthis \,,
\end{align*}

\begin{align*}
\overline{y}_{-i}^3
&= \overline{y}_{-i} \; \overline{y}_{-i}^2
    = \frac{1}{n-1} \left( y_j + A \right) \frac{1}{(n-1)^2} \left(y_j^2 + 2 y_j A + B + C\right) \\
&= \frac{1}{(n-1)^3} \left(y_j^3 + 3 y_j^2 A + y_j B + y_j C + 2 y_j A^2 + A B + A C\right)
\numberthis \,,\\
\E\left[\overline{y}_{-i}^3\right]
&= \frac{1}{(n-1)^3} \Big(\E[y_i^3] + 3 \E[y_i^2] \E[A] + \E[y_i] \E[B] + \E[y_i] \E[C] \\
    &\qquad\qquad + 2 \E[y_i] \E[A^2] + \E[A B] + \E[A C]\Big)
\\
&= \frac{1}{(n-1)^3} \Bigg(
    \E[y_i^3]
    + 3 \E[y_i^2] (n-2)\E[y_i]
    + \E[y_i] (n-2)\E[y_i^2] \\ &\qquad\qquad
    + \E[y_i] (n-2)(n-3)\E[y_i]^2 \\ &\qquad\qquad
    + 2 \E[y_i] \left((n-2)\E[y_i^2] + (n-2)(n-3)\E[y_i]^2\right) \\ &\qquad\qquad
    + \left((n-2)\E[y_i^3] + (n-2)(n-3)\E[y_i]\E[y_i^2]\right) \\ &\qquad\qquad
    + (n-2)(n-3)\left( 2 \E[y_i^2]\E[y_i] + (n-4)\E[y_i]^3 \right)
    \Bigg)
\\
&= \frac{1}{(n-1)^2} \bigg(
    \E[y_i^3]
    + 3 (n-2) \E[y_i] \E[y_i^2]
    + (n-2)(n-3)\E[y_i]^3
    \bigg)
\\
&= \frac{1}{(n-1)^2} \bigg(
    \mu_3 + 3 \sigma^2 \mu + \mu^3
    + 3 (n-2) \mu \left(\mu^2 + \sigma^2\right)
    + (n-2)(n-3)\mu^3
    \bigg)
\\
&= \mu^3 + \frac{1}{(n-1)^2} \mu_3 + \frac{3}{n-1} \mu \sigma^2
\numberthis \,,
\end{align*}

\begin{align*}
\overline{y}_{-i}^4
&= \left( \overline{y}_{-i}^2 \right)^2
    = \left( \frac{1}{(n-1)^2} \left(y_j^2 + 2 y_j A + B + C\right) \right)^2 \\
&= \frac{1}{(n-1)^4} \bigg(
    y_j^4 + 4 y_j^2 A^2 + B^2 + C^2 \\
    &\qquad\qquad
    + 4 y_j^3 A + 2 y_j^2 B + 2 y_j^2 C
    + 4 y_j A B + 4 y_j A C
    + 2 B C
\bigg)
\numberthis \,,\\
\end{align*}
\begin{gather*}
\E\left[\overline{y}_{-i}^4\right]
= \frac{1}{(n-1)^4} \bigg(
    \E[y_i^4] + 4 \E[y_i^2]\E[A^2] + \E[B^2] + \E[C^2] \\ \qquad\qquad
    + 4 \E[y_i^3] \E[A] + 2 \E[y_i^2] \E[B] + 2 \E[y_i^2] \E[C] \\ \qquad\qquad
    + 4 \E[y_i] \E[A B] + 4 \E[y_i] \E[A C]
    + 2 \E[B C]
\bigg) \\
\qquad= \frac{1}{(n-1)^4} \Bigg( \\ \qquad\qquad
    \E[y_i^4] \\ \qquad\qquad
    + 4 \E[y_i^2] \left( (n-2)\E[y_i^2] + (n-2)(n-3)\E[y_i]^2 \right) \\ \qquad\qquad
    + \left( (n-2)\E[y_i^4] + (n-2)(n-3)\E[y_i^2]^2 \right) \\ \qquad\qquad
    + (n-2)(n-3)\bigg(2\Big(\E[y_i^2]^2 + 2 (n-4) \E[y_i^2]\E[y_i]^2 \Big) + (n-4)(n-5)\E[y_i]^4 \bigg) \\ \qquad\qquad
    + 4 \E[y_i^3] (n-2)\E[y_i]  \\ \qquad\qquad
    + 2 \E[y_i^2] (n-2)\E[y_i^2]  \\ \qquad\qquad
    + 2 \E[y_i^2] (n-2)(n-3)\E[y_i]^2 \\ \qquad\qquad
    + 4 \E[y_i] \left( (n-2)\E[y_i^3] + (n-2)(n-3)\E[y_i]\E[y_i^2] \right) \\ \qquad\qquad
    + 4 \E[y_i] (n-2)(n-3)\left( 2 \E[y_i^2]\E[y_i] + (n-4)\E[y_i]^3 \right) \\ \qquad\qquad
    + 2 (n-2)(n-3)\left( 2 \E[y_i^3]\E[y_i] + (n-4)\E[y_i]^2\E[y_i^2] \right)
\Bigg) \\
\qquad= \frac{1}{(n-1)^3} \bigg(
    \E[y_i^4]
    + 3 (n-2) \E[y_i^2]^2
    + 4 (n-2) \E[y_i]\E[y_i^3]  \\ \qquad\qquad
    + 6 (n-2)(n-3) \E[y_i]^2\E[y_i^2]
    + (n-2)(n-3)(n-4) \E[y_i]^4
\bigg) \\
\qquad= \frac{1}{(n-1)^3} \Bigg( \\ \qquad\qquad
    \left(\mu_4 + 4 \mu_3 \mu + 6 \sigma^2 \mu^2 + \mu^4\right) \\ \qquad\qquad
    + 3 (n-2) \left(\mu^4 + 2 \mu^2 \sigma^2 + \sigma^4\right) \\ \qquad\qquad
    + 4 (n-2) \mu \left(\mu_3 + 3 \sigma^2 \mu + \mu^3\right) \\ \qquad\qquad
    + 6 (n-2)(n-3) \mu^2 \left(\mu^2 + \sigma^2\right)  \\ \qquad\qquad
    + (n-2)(n-3)(n-4) \mu^4
\Bigg) \\
\qquad=
    \mu^4
    + \frac{1}{(n-1)^3} \mu_4
    + \frac{4}{(n-1)^2} \mu \mu_3
    + \frac{6}{n-1} \mu^2 \sigma^2
    + \frac{3(n-2)}{(n-1)^3} \sigma^4
\numberthis \,,
\end{gather*}

\begin{align*}
\overline{y}_{-i}^2 \overline{y}_{-j}^2 = \frac{1}{(n-1)^4} \bigg(
&&& + y_i^2 y_j^2 && + y_i^2 2 y_j A && + y_i^2 B && + y_i^2 C && \\
&&& + y_j^2 2 y_i A && + 2 y_i A 2 y_j A && + 2 y_i A B && + 2 y_i A C && \\
&&& + y_j^2 B && + 2 y_j A B && + B^2 && + B C && \\
&&& + y_j^2 C && + 2 y_j A C && + B C && + C^2 && \bigg) \numberthis  \,,
\end{align*}
\begin{gather*}
\E\left[\overline{y}_{-i}^2 \overline{y}_{-j}^2\right] = \frac{1}{(n-1)^4} \bigg(
    \E[y_i^2 y_j^2] + \E[4 y_i y_j A^2] + \E[B^2] + \E[C^2] \\
    \qquad\qquad + 2 \Big( \E[2 y_i y_j^2 A] + \E[y_i^2 B] + \E[y_i^2 C] \\
    \qquad\qquad + \E[2 y_i A B] + \E[2 y_i A C] + \E[B C] \Big) \bigg) \\
\qquad= \frac{1}{(n-1)^4} \bigg(
    \E[y_i^2]^2 + 4 \E[y_i]^2 \E[A^2] + \E[B^2] + \E[C^2] \\
    \qquad\qquad + 2 \Big( 2 \E[y_i] \E[y_i^2] \E[A] + \E[y_i^2] \E[B] + \E[y_i^2] \E[C] \\
    \qquad\qquad + 2 \E[y_i] \E[A B] + 2\E[y_i] \E[A C] + \E[B C] \Big) \bigg) \\
\qquad= \frac{1}{(n-1)^4} \Bigg( \\
    \qquad\qquad \E[y_i^2]^2 + 4 \E[y_i]^2 \left( (n-2)\E[y_i^2] + (n-2)(n-3)\E[y_i]^2 \right) \\
    \qquad\qquad + \left( (n-2)\E[y_i^4] + (n-2)(n-3)\E[y_i^2]^2 \right) \\
    \qquad\qquad + (n-2)(n-3)\bigg(2\Big(\E[y_i^2]^2 + 2 (n-4) \E[y_i^2]\E[y_i]^2 \Big) + (n-4)(n-5)\E[y_i]^4 \bigg) \\
    \qquad\qquad + 2 \bigg( 2 \E[y_i] \E[y_i^2] (n-2)\E[y_i] + \E[y_i^2] (n-2)\E[y_i^2] + \E[y_i^2] (n-2)(n-3)\E[y_i]^2 \\
    \qquad\qquad + 2 \E[y_i] \left( (n-2)\E[y_i^3] + (n-2)(n-3)\E[y_i]\E[y_i^2] \right) \\
    \qquad\qquad+ 2\E[y_i]  (n-2)(n-3)\left( 2 \E[y_i^2]\E[y_i] + (n-4)\E[y_i]^3 \right) \\
    \qquad\qquad+ (n-2)(n-3)\left( 2 \E[y_i^3]\E[y_i] + (n-4)\E[y_i]^2\E[y_i^2] \right) \bigg) \Bigg) \\
\qquad= \frac{1}{(n-1)^4} \Bigg( \\
    \qquad\qquad + \Big( 4(n-2)(n-3) +(n-2)(n-3)(n-4)(n-5) +4(n-2)(n-3)(n-4) \Big) \E[y_i]^4 \\
    \qquad\qquad + \Big(1 + (n-2)(n-3) + 2(n-2)(n-3) + 2(n-2) \Big) \E[y_i^2]^2 \\
    \qquad\qquad + \Big( 4(n-2) + 4(n-2)(n-3)(n-4) + 4(n-2) + 2(n-2)(n-3) \\
    \qquad\qquad \qquad +4(n-2)(n-3) + 8(n-2)(n-3) + 2(n-2)(n-3)(n-4) \Big) \E[y_i]^2\E[y_i^2] \\
    \qquad\qquad + \Big( 4(n-2) + 4(n-2)(n-3) \Big) \E[y_i^3]\E[y_i] \\
    \qquad\qquad + (n-2) \E[y_i^4]
\Bigg) \\
\qquad= \frac{1}{(n-1)^4} \Bigg( \\
    \qquad\qquad + (n-2)(n-3)\left(n^2-5n+8\right) \mu^4 \\
    \qquad\qquad + \left(3n^2 - 13n + 15 \right) \left( \mu^2 + \sigma^2 \right)^2 \\
    \qquad\qquad +  2(n-2)\left(3n^2-14n+19\right) \mu^2\left( \mu^2 + \sigma^2 \right) \\
    \qquad\qquad + 4(n-2)^2 \mu \left( \mu_3 + 3 \sigma^2 \mu + \mu^3 \right) \\
    \qquad\qquad + (n-2) \left( \mu_4 + 4 \mu_3 \mu + 6 \sigma^2 \mu^2 + \mu^4 \right)
\Bigg) \\
\qquad= \frac{1}{(n-1)^4} \Bigg(
    \; (n-1)^4 \mu^4
    \; + 2(n-1)^2(3n-5) \sigma^2 \mu^2
    \; + \left(3n^2 - 13n + 15 \right) \sigma^4 \\
    \qquad\qquad + 4(n-2)(n-1) \mu_3 \mu
    \; + (n-2) \mu_4
\Bigg) \\
\qquad= \mu^4
    + 2\frac{3n-5}{(n-1)^2} \sigma^2 \mu^2
    + \frac{3n^2 - 13n + 15}{(n-1)^4} \sigma^4
    + 4\frac{n-2}{(n-1)^3} \mu_3 \mu
    + \frac{n-2}{(n-1)^4} \mu_4
\numberthis \,,
\end{gather*}

\begin{align*}
y_i \overline{y}_{-i} \; \overline{y}_{-j}^2
    &= y_i \frac{1}{n-1} \left( y_j + A \right) \frac{1}{(n-1)^2} \left(y_i^2 + 2 y_i A + B + C\right) \\
    &= \frac{1}{(n-1)^3} \left( y_i^3 y_j + 2 y_i^2 y_j A + y_i y_j B + y_i y_j C + y_i^3 A + 2 y_i^2 A^2 + y_i A B + y_i A C  \right)
    \numberthis \,,
\end{align*}
\begin{gather*}
\E\left[y_i \overline{y}_{-i} \; \overline{y}_{-j}^2\right]
= \frac{1}{(n-1)^3} \bigg(
    \E[y_i^3] \E[y_i] + 2 \E[y_i^2] \E[y_i] \E[A] + \E[y_i]^2 \E[B] + \E[y_i]^2 \E[C] \\
    \qquad\qquad+ \E[y_i^3] \E[A] + 2 \E[y_i^2] \E[A^2] + \E[y_i] \E[A B] + \E[y_i] \E[A C]
\bigg) \\
\qquad= \frac{1}{(n-1)^3} \bigg(
    \E[y_i^3] \E[y_i] + 2 \E[y_i^2] \E[y_i] (n-2)\E[y_i] + \E[y_i]^2 (n-2)\E[y_i^2] \\
    \qquad\qquad+ \E[y_i]^2 (n-2)(n-3)\E[y_i]^2 + \E[y_i^3] (n-2)\E[y_i] \\
    \qquad\qquad+ 2 \E[y_i^2] \left((n-2)\E[y_i^2]+(n-2)(n-3)\E[y_i]^2\right) \\
    \qquad\qquad+ \E[y_i] \left((n-2)\E[y_i^3]+(n-2)(n-3)\E[y_i]\E[y_i^2]\right) \\
    \qquad\qquad+ \E[y_i] (n-2)(n-3)\left( 2 \E[y_i^2]\E[y_i] + (n-4)\E[y_i]^3 \right)
\bigg) \\
\qquad= \frac{1}{(n-1)^3} \bigg(
    (n-2)(n-3)^2 \E[y_i]^4
    + 2(n-2) \E[y_i^2]^2 \\
    \qquad\qquad
    + (n-2)(5n-12) \E[y_i]^2\E[y_i^2]
    + (2n-3)\E[y_i]\E[y_i^3]
\bigg) \\
\qquad= \frac{1}{(n-1)^3} \bigg(
    (n-2)(n-3)^2 \mu^4
    + 2(n-2) \left(\mu^4 + 2 \mu^2 \sigma^2 + \sigma^4\right) \\
    \qquad\qquad
    + (n-2)(5n-12) \mu^2\left(\mu^2 + \sigma^2\right)
    + (2n-3)\mu\left(\mu_3 + 3 \sigma^2 \mu + \mu^3\right)
\bigg) \\
\qquad= \frac{1}{(n-1)^3} \bigg(
    (n-1)^3 \mu^4
    + (2n-3) \mu \mu_3
    + (n-1)(5n-7) \mu^2 \sigma^2
    +2(n-2) \sigma^4
\bigg) \\
\qquad=
    \mu^4
    + \frac{2n-3}{(n-1)^3} \mu \mu_3
    + \frac{5n-7}{(n-1)^2} \mu^2 \sigma^2
    + \frac{2(n-2)}{(n-1)^3} \sigma^4
\numberthis \,,
\end{gather*}

\begin{gather*}
y_i^2 \overline{y}_{-j}
    = y_i^2 \frac{1}{n-1} \left( y_i + A \right)
    = \frac{1}{n-1} \left( y_i^3 + y_i^2 A \right)
\numberthis \,,
\end{gather*}
\begin{align*}
\E\left[y_i^2 \overline{y}_{-j}\right]
    &= \frac{1}{n-1} \left( \E[y_i^3] + \E[y_i^2] \E[A] \right) \\
&=\frac{1}{n-1} \left(
    \left(\mu_3 + 3 \sigma^2 \mu + \mu^3\right)
    + \left(\mu^2 + \sigma^2\right) (n-2)\mu \right) \\
&= \mu^3 + \frac{1}{n-1} \mu_3 + \frac{n+1}{n-1}\mu\sigma^2
\numberthis \,,
\end{align*}

\begin{align*}
y_i^2 \overline{y}_{-j}^2
    &= y_i^2 \frac{1}{(n-1)^2} \left(y_i^2 + 2 y_i A + B + C\right)
    \\
    &= \frac{1}{(n-1)^2} \left( y_i^4 + 2 y_i^3 A +y_i^2 B + y_i^2 C  \right)
    \numberthis \,,
\end{align*}
\begin{align*}
\E\left[y_i^2 \overline{y}_{-j}^2\right]
    &= \frac{1}{(n-1)^2} \left( \E[y_i^4] + 2 \E[y_i^3] \E[A] + \E[y_i^2] \E[B] + \E[y_i^2] \E[C]  \right) \\
&=\frac{1}{(n-1)^2} \left(
    \E[y_i^4] + 2(n-2) \E[y_i]\E[y_i^3] + (n-2)\E[y_i^2]^2 + (n-2)(n-3)\E[y_i^2]\E[y_i]^2
    \right) \\
&=\frac{1}{(n-1)^2} \bigg(
    \left(\mu_4 + 4 \mu_3 \mu + 6 \sigma^2 \mu^2 + \mu^4\right)
    + 2(n-2) \mu\left(\mu_3 + 3 \sigma^2 \mu + \mu^3\right)\\
    &\qquad\qquad+ (n-2)\left(\mu^4 + 2  \mu^2 \sigma^2 + \sigma^4\right)
    + (n-2)(n-3)\left(\mu^2 + \sigma^2\right)\mu^2
    \bigg) \\
&=
    \mu^4
    + \frac{1}{(n-1)^2}\mu_4
    + \frac{2n}{(n-1)^2} \mu \mu_3
    + \frac{n+4}{n-1} \mu^2 \sigma^2
    + \frac{n-2}{(n-1)^2} \sigma^4
    \numberthis \,.
\end{align*}

\endgroup 
\makeatletter
\setbool{@fleqn}{false}
\makeatother

\noindent
Consider Bayesian normal model $\Model$ with known variance $\sigma^2_\mathrm{m}$ and prior $\operatorname{N}\left(0, \sigma_0^2\right)$:
\begin{align}
    y \mid \theta, \Model \sim \operatorname{N}(\theta, \sigma^2_\text{m}), \qquad \theta \sim \operatorname{N}(0, \sigma_0^2) \,.
\end{align}
The posterior predictive distribution is \citep[see e.g.][pp.\ 39--42]{bda_book}
\begin{gather}
    \widetilde{y} \mid y, \Model \sim \operatorname{N}\left(\mu_\text{pp}(y) \,, \sigma^2_\text{pp}(n) \right) \,,\\
\intertext{where}
    \mu_\text{pp}(y) = \tau(n) \frac{n}{\sigma^2_\mathrm{m}}\overline{y} \,,\qquad
    \sigma^2_\text{pp}(n) = \sigma^2_\mathrm{m} + \tau(n) \,, \\
    \tau(n) = \left( \frac{1}{\sigma_0^2} + \frac{n}{\sigma^2_\mathrm{m}} \right)^{-1}
    \,.
\end{gather}
The log predictive density is
\begin{align*}
    \log \p(\widetilde{y}|y)
    &= -\left(2\sigma^2_\text{pp}(n)\right)^{-1}
    \left( \widetilde{y} - \tau(n) \frac{n}{\sigma^2_\mathrm{m}}\overline{y} \right)^2
    - \frac{1}{2}\log\left( 2 \pi \sigma^2_\text{pp}(n) \right) \\
    &= a(n) \widetilde{y}^2 + b(n) \widetilde{y}\overline{y} + c(n) \overline{y}^2 + d(n)
    \,,\numberthis
\end{align*}
where
\begingroup
\allowdisplaybreaks
\begin{align*}
    a(n) &= - \frac{1}{2 \sigma^2_\text{pp}(n)} \\
        &= - \frac{1}{2} \frac{ \sigma_\mathrm{m}^2 + n \sigma_0^2}{\sigma_\mathrm{m}^2(\sigma_\mathrm{m}^2 + (n+1) \sigma_0^2)}
        \,, \numberthis \\
    b(n) &= \frac{\tau(n)}{\sigma^2_\mathrm{m} \sigma^2_\text{pp}(n)} n \\
        &= \frac{n \sigma_0^2}{\sigma_\mathrm{m}^2 (\sigma_\mathrm{m}^2 + (n+1) \sigma_0^2)}
        \,, \numberthis \\
    c(n) &= - \frac{\tau(n)^2}{2 \sigma^4_\mathrm{m} \sigma^2_\text{pp}(n)} n^2 \\
        &= - \frac{1}{2} \frac{n^2 \sigma_0^4}{\sigma_\mathrm{m}^2(\sigma_\mathrm{m}^2+n\sigma_0^2)(\sigma_\mathrm{m}^2+(n+1) \sigma_0^2)}
        \,, \numberthis \\
    d(n) &= - \frac{1}{2}\log\left( 2 \pi \sigma^2_\text{pp}(n) \right) \\
        &= - \frac{1}{2} \log\left( 2\pi \frac{\sigma_\mathrm{m}^2(\sigma_\mathrm{m}^2 + (n+1) \sigma_0^2)}{\sigma_\mathrm{m}^2 + n \sigma_0^2} \right)
        \,. \numberthis
\end{align*}
\endgroup
Terms $a(n)$, $b(n)$, $c(n)$, $d(n)$, $\sigma^2_\text{pp}(n)$, and $\tau(n)$ do not depend on the individual observations $y_i$, but only on the number of observations.
The LOO-CV pointwise predictive performance estimate for observation $i$ using elpd utility measure is
\begin{gather}
\elpdHati{\Model}{y}{i} = \log \p(y_i|y_{-i}) = a(n-1) y_i^2 + b(n-1) y_i \overline{y}_{-i} + c(n-1) \overline{y}_{-i}^2 + d(n-1)
\end{gather}
In the following, $a$, $b$, $c$, and, $d$ are notated without the function argument $n-1$ and we use short notation $\ShortElpdHat_i \coloneqq \elpdHati{\Model}{y}{i}$ in order to make the notation more compact.
We have the following expectations:
\makeatletter
\setbool{@fleqn}{true}
\makeatother
\begingroup
\allowdisplaybreaks
\begin{align*}
\E[\ShortElpdHat_i] &= a \E[y_i^2] + b \E[y_i] \E[\overline{y}_{-i}] + c \E[\overline{y}_{-i}^2] + d \\
    &= a (\mu^2 + \sigma^2) + b \mu \mu + c \left( \mu^2 + \frac{1}{n-1} \sigma^2 \right) + d \\
    &= (a + b + c) \mu^2 + \left( a + \frac{1}{n-1}c \right)\sigma^2 + d
    \numberthis \,,\\
\E[\ShortElpdHat_i]^2 &=
    (a + b + c)^2 \mu^4
    + \left(a+\frac{1}{n-1}c\right)^2\sigma^4
    + d^2
    \\
    &\qquad
    + 2 (a + b + c) \left(a+\frac{1}{n-1}c\right) \mu^2 \sigma^2
    + 2 d (a + b + c) \mu^2
    + 2 d \left(a+\frac{1}{n-1}c\right)\sigma^2
    \numberthis \,,
\end{align*}

\begin{align*}
\ShortElpdHat_i^2=&
    && + a^2 y_i^4
    && + b^2 y_i^2 \overline{y}_{-i}^2
    && + c^2 \overline{y}_{-i}^4
    && + d^2 \\
&
    && + 2 a b y_i^3 \overline{y}_{-i}
    && + 2 a c y_i^2 \overline{y}_{-i}^2
    && + 2 a d y_i^2
    &&  \\
&
    && + 2 b c y_i \overline{y}_{-i}^3
    && + 2 b d y_i \overline{y}_{-i}
    &&
    && \\
&
    && + 2 c d \overline{y}_{-i}^2 \,,
    &&
    &&
    &&
    \numberthis
\end{align*}
\begin{align*}
\E[\ShortElpdHat_i^2]=&
    && + a^2 \E[y_i^4]
    && + b^2 \E[y_i^2] \E[\overline{y}_{-i}^2]
    && + c^2 \E[\overline{y}_{-i}^4]
    && + d^2 \\
&
    && + 2 a b \E[y_i^3]\E[\overline{y}_{-i}]
    && + 2 a c \E[y_i^2]\E[\overline{y}_{-i}^2]
    && + 2 a d \E[y_i^2]
    &&  \\
&
    && + 2 b c \E[y_i]\E[\overline{y}_{-i}^3]
    && + 2 b d \E[y_i]\E[\overline{y}_{-i}]
    &&
    && \\
&
    && + 2 c d \E[\overline{y}_{-i}^2]
    &&
    &&
    &&
\end{align*}
\begin{gather*}
\qquad =
    a^2 \left(\mu_4 + 4 \mu_3 \mu + 6 \sigma^2 \mu^2 + \mu^4\right) \\ \qquad\qquad
    + b^2 \left(\mu^2 + \sigma^2\right) \left(\mu^2 + \frac{1}{n-1} \sigma^2\right) \\ \qquad\qquad
    + c^2 \left(
        \mu^4
        + \frac{1}{(n-1)^3} \mu_4
        + \frac{4}{(n-1)^2} \mu \mu_3
        + \frac{6}{n-1} \mu^2 \sigma^2
        + \frac{3(n-2)}{(n-1)^3} \sigma^4\right) \\ \qquad\qquad
    + d^2 \\ \qquad\qquad
    + 2 a b \left(\mu_3 + 3 \sigma^2 \mu + \mu^3\right) \mu \\ \qquad\qquad
    + 2 a c \left(\mu^2 + \sigma^2\right) \left(\mu^2 + \frac{1}{n-1} \sigma^2\right) \\ \qquad\qquad
    + 2 a d \left(\mu^2 + \sigma^2\right) \\ \qquad\qquad
    + 2 b c \mu \left(\mu^3 + \frac{1}{(n-1)^2} \mu_3 + \frac{3}{n-1} \mu \sigma^2\right) \\ \qquad\qquad
    + 2 b d \mu \mu \\ \qquad\qquad
    + 2 c d \left(\mu^2 + \frac{1}{n-1} \sigma^2\right) \\
\qquad=
    (a+b+c)^2 \mu^4 \\ \qquad\qquad
    + 2 d \left(a + b + c\right) \mu^2 \\ \qquad\qquad
    + \left(
        6 a^2
        + \frac{n}{n-1} b^2
        + \frac{6}{n-1} c^2
        + 6 a b
        + \frac{2n}{n-1} a c
        + \frac{6}{n-1} b c
        \right) \mu^2 \sigma^2 \\ \qquad\qquad
    + \left(
        \frac{1}{n-1} b^2
        + \frac{3(n-2)}{(n-1)^3} c^2
        + \frac{2}{n-1} a c
        \right) \sigma^4 \\ \qquad\qquad
    + 2 d \left( a + \frac{1}{n-1} c \right) \sigma^2 \\ \qquad\qquad
    + \left(
        4 a^2
        + 2 a b
        + \frac{2}{(n-1)^2} b c
        + \frac{4}{(n-1)^2} c^2
        \right) \mu \mu_3 \\ \qquad\qquad
    + \left(
        a^2
        + \frac{1}{(n-1)^3}c^2
        \right) \mu_4 \\ \qquad\qquad
    + d^2
    \numberthis \,,
\end{gather*}

\begin{align*}
\ShortElpdHat_i \ShortElpdHat_j =
&&& + a^2 y_i^2 y_j^2 && + ab y_i^2y_j\overline{y}_{-j} && + a c y_i^2 \overline{y}_{-j}^2 && + ady_i^2 \\
&&& + a b y_i \overline{y}_{-i} y_j^2 && + b^2 y_i \overline{y}_{-i}y_j \overline{y}_{-j}
    && + bc y_i \overline{y}_{-i} \, \overline{y}_{-j}^2  && + bd y_i \overline{y}_{-i} \\
&&& + ac \overline{y}_{-i}^2 y_j^2 && + bc \overline{y}_{-i}^2 y_j \overline{y}_{-j}
    && + c^2 \overline{y}_{-i}^2\overline{y}_{-j}^2 && + cd \overline{y}_{-i}^2 \\
&&& + ad y_j^2 && + bd y_j \overline{y}_{-j} && + cd \overline{y}_{-j}^2 && + d^2
\numberthis\,,
\end{align*}
\begin{align*}
\E\left[\ShortElpdHat_i \ShortElpdHat_j\right]=&
    && + a^2 \E[y_i^2 y_j^2]
    && + b^2 \E[y_i y_j \overline{y}_{-i} \; \overline{y}_{-j}]
    && + c^2 \E[\overline{y}_{-i}^2\overline{y}_{-j}^2]
    && + d^2 \\
&
    && + 2 ab \E[y_i^2y_j\overline{y}_{-j}]
    && + 2 a c \E[y_i^2 \overline{y}_{-j}^2]
    && + 2 ad\E[y_i^2]
    &&  \\
&
    && + 2 bc \E[y_i \overline{y}_{-i} \, \overline{y}_{-j}^2]
    && + 2 bd \E[y_i \overline{y}_{-i}]
    &&
    && \\
&
    && + 2 cd \E[\overline{y}_{-i}^2]
    &&
    &&
    && \\
\qquad=&
    && + a^2 \E[y_i^2]^2
    && + b^2 \E[y_i y_j \overline{y}_{-i} \; \overline{y}_{-j}]
    && + c^2 \E[\overline{y}_{-i}^2\overline{y}_{-j}^2]
    && + d^2 \\
&
    && + 2 ab \E[y_i]\E[y_i^2\overline{y}_{-j}]
    && + 2 ac \E[y_i^2 \overline{y}_{-j}^2]
    && + 2 ad\E[y_i^2]
    &&  \\
&
    && + 2 bc \E[y_i \overline{y}_{-i} \, \overline{y}_{-j}^2]
    && + 2 bd \E[y_i]\E[\overline{y}_{-i}]
    &&
    && \\
&
    && + 2 cd \E[\overline{y}_{-i}^2] \,,
    &&
    &&
    &&
    \numberthis
\end{align*}
\begin{gather*}
\qquad=
    a^2 \left( \mu^4 + 2 \mu^2 \sigma^2 + \sigma^4 \right) \\
\qquad\qquad+
    b^2 \left(\mu^4
        + \frac{3n-4}{(n-1)^2} \mu^2 \sigma^2
        + \frac{1}{(n-1)^2} \sigma^4
    \right)\\
\qquad\qquad+
    c^2 \left(\mu^4
        + 2\frac{3n-5}{(n-1)^2} \sigma^2 \mu^2
        + \frac{3n^2 - 13n + 15}{(n-1)^4} \sigma^4
        + 4\frac{n-2}{(n-1)^3} \mu_3 \mu
        + \frac{n-2}{(n-1)^4} \mu_4
    \right)\\
\qquad\qquad+
    d^2 \\
\qquad\qquad+
    2 ab \mu \left( \mu^3 + \frac{1}{n-1} \mu_3 + \frac{n+1}{n-1}\mu\sigma^2 \right) \\
\qquad\qquad+
    2 ac \left( \mu^4
    + \frac{1}{(n-1)^2}\mu_4
    + \frac{2n}{(n-1)^2} \mu \mu_3
    + \frac{n+4}{n-1} \mu^2 \sigma^2
    + \frac{n-2}{(n-1)^2} \sigma^4 \right) \\
\qquad\qquad+
    2 ad \left(\mu^2 + \sigma^2 \right) \\
\qquad\qquad+
    2 bc \left( \mu^4
    + \frac{2n-3}{(n-1)^3} \mu \mu_3
    + \frac{5n-7}{(n-1)^2} \mu^2 \sigma^2
    + \frac{2(n-2)}{(n-1)^3} \sigma^4 \right) \\
\qquad\qquad+
    2 bd \mu \mu \\
\qquad\qquad+
    2 cd \left( \mu^2 + \frac{1}{n-1} \sigma^2 \right) \\
\qquad=
    \left( a + b + c \right)^2 \mu^4 \\
\qquad\qquad+
    2 d \left( a + b + c \right) \mu^2 \\
\qquad\qquad+
    \Bigg(
    2 a^2
    + \frac{3n-4}{(n-1)^2} b^2
    + \frac{2(3n-5)}{(n-1)^2}c^2 \\ \qquad\qquad\qquad
    + \frac{2(n+1)}{n-1} a b
    + \frac{2(n+4)}{n-1} a c
    + \frac{2(5n-7)}{(n-1)^2} b c
    \Bigg) \mu^2 \sigma^2 \\
\qquad\qquad+
    \left(
    a^2
    + \frac{1}{(n-1)^2} b^2
    + \frac{3 n^2 - 13 n +15}{(n-1)^4} c^2
    + \frac{2(n-2)}{(n-1)^2} a c
    + \frac{4(n-2)}{(n-1)^3} b c
    \right) \sigma^4 \\
\qquad\qquad+
    2 d \left(a + \frac{1}{n-1} c \right) \sigma^2 \\
\qquad\qquad+
    \left(
    \frac{4(n-2)}{(n-1)^3} c^2
    + \frac{2}{n-1} a b
    + \frac{4n}{(n-1)^2} a c
    + \frac{4n-6}{(n-1)^3} b c
    \right) \mu \mu_3 \\
\qquad\qquad+
    \left(
    \frac{n-2}{(n-1)^4} c^2
    + \frac{2}{(n-1)^2} a c
    \right) \mu_4 \\
\qquad\qquad+
    d^2
    \numberthis \,.
\end{gather*}
%
\endgroup 
\makeatletter
\setbool{@fleqn}{false}
\makeatother
%
With these expectations, we can derive the variance and covariance of the pointwise LOO-CV estimates:
\begin{align*}
\Var(\ShortElpdHat_i) &= \E[\ShortElpdHat_i^2] - \E[\ShortElpdHat_i]^2 \\
&=
    \left(4a^2 + \frac{n}{n-1} b^2 + \frac{4}{n-1} c^2 + 4 a b + \frac{4}{n-1} b c\right) \mu^2 \sigma^2
    \\ &\qquad\qquad
    + \left(-a^2 + \frac{1}{n-1} b^2 + \frac{2n-5}{(n-1)^3} c^2 \right) \sigma^4
    \\ &\qquad\qquad
    + \left(4 a^2 + \frac{4}{(n-1)^2} c^2 + 2 a b + \frac{2}{(n-1)^2} b c \right) \mu \mu_3
    \\ &\qquad\qquad
    + \left(a^2 + \frac{1}{(n-1)^3} c^2\right) \mu_4
    \,,\numberthis \label{eq_app_var_looi}
\end{align*}
\begin{align*}
\Cov(\ShortElpdHat_i, \ShortElpdHat_j) &= \E[\ShortElpdHat_i \ShortElpdHat_j] - \E[\ShortElpdHat_i] \E[\ShortElpdHat_j]
    \\
    &= \E[\ShortElpdHat_i \ShortElpdHat_j] - \E[\ShortElpdHat_i]^2 \\
& =
    \left(
    \frac{3n-4}{(n-1)^2} b^2
    + \frac{4(n-2)}{(n-1)^2} c^2
    + \frac{4}{n-1} a b
    + \frac{8}{n-1} a c
    + \frac{4(2n-3)}{(n-1)^2} b c
    \right) \mu^2 \sigma^2 \\
&\qquad\qquad+
    \left(
    \frac{1}{(n-1)^2} b^2
    + \frac{(n-2)(2n-7)}{(n-1)^4} c^2
    - \frac{2}{(n-1)^2} a c
    + \frac{4(n-2)}{(n-1)^3} b c
    \right) \sigma^4 \\
&\qquad\qquad+
    \left(
    \frac{4(n-2)}{(n-1)^3} c^2
    + \frac{2}{n-1} a b
    + \frac{4n}{(n-1)^2} a c
    + \frac{4n-6}{(n-1)^3} b c
    \right) \mu \mu_3 \\
&\qquad\qquad+
    \left(
    \frac{n-2}{(n-1)^4} c^2
    + \frac{2}{(n-1)^2} a c
    \right) \mu_4
    \,.\numberthis\label{eq_app_cov_looi_looj}
\end{align*}
The variance of the sum of the pointwise LOO-CV terms $\elpdHat{\Model}{y} = \sum_{p=1}^n \ShortElpdHat_p$ is \citep{bengio_Grandvalet_2004}
\begin{align*}
    \Var\left(\elpdHat{\Model}{y}\right)
    &= n \Var(\ShortElpdHat_i) + n(n-1) \Cov(\ShortElpdHat_i, \ShortElpdHat_j)
\,.\numberthis\label{eq_app_var_elpd_bengio_granv}
\end{align*}
Combining from equations~\eqref{eq_app_var_looi}, \eqref{eq_app_cov_looi_looj}, and~\eqref{eq_app_var_elpd_bengio_granv}, we get the desired result:
\begin{align*}
    &\Var\left(\elpdHat{\Model}{y}\right)
        \\
&\qquad= 4 n (a + b + c)^2 \mu^2 \sigma^2 \\
&\qquad\qquad+
    \left(
        - n a^2
        + \frac{2n}{n-1} b^2
        + \frac{n (2 n - 3) (n - 3)}{(n - 1)^3} c^2
        - \frac{2n}{n-1} a c
        + \frac{4n(n-2)}{(n-1)^2} b c
    \right) \sigma^4 \\
&\qquad\qquad+
    \frac{4 n (a + b + c) (a (n - 1) + c)}{n - 1} \mu \mu_3 \\
&\qquad\qquad+
    \left(
        n a^2
        + \frac{n}{(n-1)^2} c^2
        + \frac{2n}{n-1} a c
    \right) \mu_4
    \,.\numberthis\label{eq_app_var_elpd}
\end{align*}
\end{proof}



\subsection{Proof of Lemma~\ref{lemma_estims_for_some_moment_products}}
\label{sec_app_samp_estims_for_terms}

In this section, we give a proof for Lemma~\ref{lemma_estims_for_some_moment_products} by showing that it is possible to construct sample based estimates for terms $\mu^2\sigma^2$, $\sigma^4$, $\mu \mu_3$, and $\mu_4$ given the data set $y$ of $n$ independent observations, where
\begin{align*}
    \E[y_i] &= \mu \,, \\
    \Var(y_i) &= \sigma^2 \,, \\
    \E[(y_i - \E[y_i])^r] &= \mu_r \quad \text{($r$th central moment)}
    \numberthis\,,
\end{align*}
and $n \geq 4$.
Let us first restate the lemma.
\begin{quote}
    Let the data generating mechanism for $y=[y_1,y_2,\dots,y_n]$ be such that Equation~\eqref{eq_assumptions} holds.
    Let
    \begin{equation}
    \widehat{\alpha}_k = \frac{1}{n}\sum_{i=1}^n y_i^k
    \end{equation}
    be the $k$th sample raw moment of the data and
    \begin{align*}
        \widehat{\mu^4} = \binom{n}{4}^{-1} \sum_{i_1 \neq i_2 \neq i_3  \neq i_4} y_{i_1}y_{i_2}y_{i_3}y_{i_4}
       \,, \numberthis
    \end{align*}
    where the summation is over all possible combinations of $i_1, i_2, i_3, i_4 \in \{1,2,\dots,n\}$ without repetition,
    be an unbiased estimator for the fourth power of the mean.
    Now
    \begingroup
    \allowdisplaybreaks
    \begin{align*}
        \widehat{\mu^2 \sigma^2} &= \frac{- n^3 \widehat{\alpha}_1^4 + 2 n^3 \widehat{\alpha}_2 \widehat{\alpha}_1^2 - 4 (n - 1) n \widehat{\alpha}_3 \widehat{\alpha}_1 - (2 n^2 - 3n) \widehat{\alpha}_2^2 + 2 (2 n - 3) \widehat{\alpha}_4}{2 (n - 3) (n - 2) (n - 1)} - \frac{1}{2} \widehat{\mu^4}
        \,. \numberthis
    \\
        \widehat{\sigma^4} &= \frac{
            n^3 \widehat{\alpha}_1^4 - 2 n^3 \widehat{\alpha}_2 \widehat{\alpha}_1^2 + (n^3 - 3 n^2 + 3n) \widehat{\alpha}_2^2 + 4 n (n - 1) \widehat{\alpha}_3 \widehat{\alpha}_1 + n (1 - n) \widehat{\alpha}_4
        }{
            (n - 3) (n - 2) (n - 1)
        }
        \,, \numberthis 
    \\
        \widehat{\mu \mu_3} &= \frac{
            - 2 (n^2 + n -3) \widehat{\alpha}_4
            - 6 n^3 \widehat{\alpha}_1^2 \widehat{\alpha}_2
            + n (6 n -9) \widehat{\alpha}_2^2
            + 3 n^3 \widehat{\alpha}_1^4
            + 2 n^2 (n + 1) \widehat{\alpha}_1 \widehat{\alpha}_3
        }{
            2 (n-3) (n-2) (n-1)
        }
        + \frac{1}{2} \widehat{\mu^4}
        \,, \numberthis
    \intertext{and}
        \widehat{\mu_4} &= \frac{
            -3n^4 \widehat{\alpha}_1^4
            +6n^4 \widehat{\alpha}_1^2\widehat{\alpha}_2
            +(9-6n)n^2 \widehat{\alpha}_2^2
            +(-12+8n-4n^2)n^2 \widehat{\alpha}_1\widehat{\alpha}_3
            +(3n-2n^2+n^3)n \widehat{\alpha}_4
        }{
            (n-3)(n-2)(n-1)n
        }
         \numberthis 
    \end{align*}
    \endgroup
    are unbiased estimators for the parameters $\mu^2 \sigma^2$, $\sigma^4$, $\mu \mu_3$, and $\mu_4$ respectively.
\end{quote}
\begin{proof}
\citet{unbiased_mean_powers} present an unbiased estimator for $\mu^4$:
\begin{align*}
    \widehat{\mu^4} = \frac{(n-4)!}{n!} \sum_{i_1 \neq i_2 \neq i_3  \neq i_4} y_{i_1}y_{i_2}y_{i_3}y_{i_4}
    \numberthis\,,
\end{align*}
where the summation is over all possible permutations of all possible sets of $i_1, i_2, i_3, i_4 \in \{1,2,\dots,n\}$ without repetition.
For efficiency, we suggest using the following form of this estimator:
\begin{align*}
    \widehat{\mu^4} = \binom{n}{4}^{-1} \sum_{i_1 \neq i_2 \neq i_3  \neq i_4} y_{i_1}y_{i_2}y_{i_3}y_{i_4}
   \,, \numberthis
\end{align*}
where the summation is over all possible combinations of $i_1, i_2, i_3, i_4 \in \{1,2,\dots,n\}$ without repetition.
\citet{espejo2013optimal} directly present an unbiased estimator for the fourth central moment $\mu_4$:
\begin{align}
    \widehat{\mu_4} &= \frac{
        -3n^4 \widehat{\alpha}_1^4
        +6n^4 \widehat{\alpha}_1^2\widehat{\alpha}_2
        +(9-6n)n^2 \widehat{\alpha}_2^2
        +(-12+8n-4n^2)n^2 \widehat{\alpha}_1\widehat{\alpha}_3
        +(3n-2n^2+n^3)n \widehat{\alpha}_4
    }{
        (n-3)(n-2)(n-1)n
    } \,,
     \label{eq_app_estim_m4}
\end{align}
\begin{align}
    \E\left[\widehat{\mu_4}\right] &= \mu_4 \,,
\intertext{where}
\widehat{\alpha}_k &= \frac{1}{n}\sum_{i=1}^n y_i^k
\end{align}
is the $k$th sample raw moment of the data.
In addition, as an auxiliary result, they present an estimator with an expectation of $\mu_4 + 3\sigma^4$ in equations~13 and~15:
\begin{align}
    t &= \frac{n}{n-1} \left( \widehat{\alpha}_4 - 4 \widehat{\alpha}_3 \widehat{\alpha}_1 + 3 \widehat{\alpha}_2^2 \right)
    \,, \\
    \E\left[t\right] &= \mu_4 + 3\sigma^4 \,.
\end{align}
It is possible to construct an unbiased estimator for $\sigma^4$ as a linear combination from these two estimators:
\begin{align*}
    \widehat{\sigma^4} &= \frac{1}{3} t - \frac{1}{3} \widehat{\mu_4}
    \\
    &= \frac{
        n^3 \widehat{\alpha}_1^4 - 2 n^3 \widehat{\alpha}_2 \widehat{\alpha}_1^2 + (n^3 - 3 n^2 + 3n) \widehat{\alpha}_2^2 + 4 n (n - 1) \widehat{\alpha}_3 \widehat{\alpha}_1 + n (1 - n) \widehat{\alpha}_4
    }{
        (n - 3) (n - 2) (n - 1)
    }
    \,, \numberthis \label{eq_app_estim_s4}
    \\
    \E\left[\widehat{\sigma^4}\right] &= \frac{1}{3} \left( \mu_4 + 3\sigma^4 \right) - \frac{1}{3} \mu_4 = \sigma^4
    \,.  \numberthis
\end{align*}
It is known that
\begin{align*}
\E[y_i^4] &= \mu_4 + 4 \mu_3 \mu + 6 \sigma^2 \mu^2 + \mu^4 \,,  \numberthis \\
\E[y_i^2]^2 &= \left( \mu^2 + \sigma^2 \right)^2 = \mu^4 + 2 \sigma^2 \mu^2 + \sigma^4
    \numberthis\,.
\end{align*}
We have
\begin{align*}
\E[\widehat{\alpha}_4]
    &= \frac{1}{n}\sum_{i=1}^n \E[y_i^4]
    \\
    &= \E[y_i^4]
    \\
    &= \mu_4 + 4 \mu_3 \mu + 6 \sigma^2 \mu^2 + \mu^4
    \,,\numberthis \label{eq_app_mu4_mu3mu_sigma2mu2_mu4_estim}
\\
\E[\widehat{\alpha}_2^2]
    &= \E\left[\frac{1}{n^2}\sum_{i} y_i^2 \sum_{j} y_j^2\right]
    = \frac{1}{n^2} \E\left[\sum_{i} y_i^4 + \sum_{i} \sum_{j\neq \{i\}} y_i^2 y_j^2\right]\\
    &= \frac{1}{n} \E[y_i^4] + \frac{n-1}{n} \E[y_i^2]^2
    \,,\numberthis
\\
\E\left[\frac{n}{n-1} \widehat{\alpha}_2^2 - \frac{1}{n-1} \widehat{\alpha}_4\right]
    &= \E[y_i^2]^2
    \\
    &= \mu^4 + 2 \sigma^2 \mu^2 + \sigma^4
    \,.\numberthis\label{eq_app_mu4_sigma2mu2_sigma4_estim}
\end{align*}
By linearly combining an unbiased estimator for $\mu^4$ and for $\sigma^4$ to the estimator presented in Equation~\eqref{eq_app_mu4_sigma2mu2_sigma4_estim}, it is possible to construct an unbiased estimate for $\sigma^2 \mu^2$:
\begin{align*}
    \widehat{\sigma^2 \mu^2} &= \frac{1}{2(n-1)} \left( n \widehat{\alpha}_2^2 - \widehat{\alpha}_4 \right) - \frac{1}{2} \widehat{\mu^4} - \frac{1}{2} \widehat{\sigma^4}
    \\
     &= \frac{- n^3 \widehat{\alpha}_1^4 + 2 n^3 \widehat{\alpha}_2 \widehat{\alpha}_1^2 - 4 (n - 1) n \widehat{\alpha}_3 \widehat{\alpha}_1 - (2 n^2 - 3n) \widehat{\alpha}_2^2 + 2 (2 n - 3) \widehat{\alpha}_4}{2 (n - 3) (n - 2) (n - 1)} - \frac{1}{2} \widehat{\mu^4}
    \,, \numberthis
    \label{eq_app_estim_sm}
    \\
     \E\left[\widehat{\sigma^2 \mu^2}\right] &= \frac{1}{2} \left(\mu^4 + 2 \sigma^2 \mu^2 + \sigma^4\right) - \frac{1}{2} \mu^4 - \frac{1}{2} \sigma^4  = \sigma^2 \mu^2
    \,. \numberthis
\end{align*}
Further, by combining the unbiased estimators for $\mu_4$, $\mu^4$, and $\sigma^2 \mu^2$ to the estimator presented in Equation~\eqref{eq_app_mu4_mu3mu_sigma2mu2_mu4_estim}, it is possible to construct an unbiased estimator for $\mu_3 \mu$:
\begin{align*}
    \widehat{\mu_3 \mu} &= \frac{1}{4} \widehat{\alpha}_4 - \frac{1}{4} \widehat{\mu^4} - \frac{1}{4} \widehat{\mu_4} - \frac{6}{4} \widehat{\sigma^2 \mu^2}
    \\
     &= \frac{
        - 2 (n^2 + n -3) \widehat{\alpha}_4
        - 6 n^3 \widehat{\alpha}_1^2 \widehat{\alpha}_2
        + n (6 n -9) \widehat{\alpha}_2^2
        + 3 n^3 \widehat{\alpha}_1^4
        + 2 n^2 (n + 1) \widehat{\alpha}_1 \widehat{\alpha}_3
    }{
        2 (n-3) (n-2) (n-1)
    }
    + \frac{1}{2} \widehat{\mu^4}
    \,, \numberthis
    \label{eq_app_estim_mm}
    \\
    \E\left[\widehat{\mu_3 \mu}\right] &= \frac{1}{4} \left(\mu_4 + 4 \mu_3 \mu + 6 \sigma^2 \mu^2 + \mu^4\right) - \frac{1}{4} \mu^4 - \frac{1}{4} \mu_4 - \frac{6}{4} \sigma^2 \mu^2  = \mu_3 \mu
    \,. \numberthis
\end{align*}
We have show that it is possible to construct an unbiased estimators for all the desired products of moments.
The presented estimators serve as an example and other possibly more optimal estimators might be constructed.
\end{proof}


\end{document}